\documentclass[aps,prl,reprint,twocolumn, longbibliography]{revtex4-1}

\usepackage{graphicx}
\usepackage{amsmath}
\usepackage{amssymb}
\usepackage{amsthm}
\usepackage{thm-restate}
\usepackage{verbatim}
\usepackage{braket}
\usepackage[vcentermath]{youngtab}
\usepackage{tikz}
\usepackage[framemethod=tikz]{mdframed}
\usepackage{complexity}

\usepackage[margin=1in]{geometry}
\usepackage{subcaption}
\usepackage[utf8]{inputenc}

\usepackage{algorithm}
\usepackage{algpseudocode}
\usepackage[colorlinks = true]{hyperref}

\usepackage{xcolor}
\definecolor{darkred}  {rgb}{0.5,0,0}
\definecolor{darkblue} {rgb}{0,0,0.5}
\definecolor{darkgreen}{rgb}{0,0.5,0}

\hypersetup{
  urlcolor   = blue,         
  linkcolor  = darkblue,     
  citecolor  = darkgreen,    
  filecolor  = darkred       
}

\newcommand{\be}{\begin{equation}}
\newcommand{\ee}{\end{equation}}
\newcommand{\ba}{\begin{array}}
\newcommand{\ea}{\end{array}}
\newcommand{\bea}{\begin{eqnarray}}
\newcommand{\eea}{\end{eqnarray}}

\newcommand{\trace}[1]{{\mathrm{Tr}{#1}}}

\newclass{\FBPP}{FBPP}
\newclass{\QAPC}{QXC}

\newtheorem{lemma}{Lemma}

\newtheorem{theorem}{Theorem}
\newtheorem{claim}{Claim}
\newtheorem{problem}{Problem}
\newtheorem{corol}{Corollary}

\begin{document}
\title{Quantum complexity of the Kronecker coefficients}
\author{Sergey Bravyi$^1$}
\author{Anirban Chowdhury$^{2,3}$}
\author{David Gosset$^{2,3,4}$}
\author{Vojt\v{e}ch Havl\'{i}\v{c}ek$^{1}$}
\author{Guanyu Zhu$^1$}
\affiliation{$^1$ IBM Quantum, IBM T.J. Watson Research Center}
\affiliation{$^2$ Department of Combinatorics and Optimization, University of Waterloo}
\affiliation{$^3$ Institute for Quantum Computing, University of Waterloo}
\affiliation{$^4$ Perimeter Institute for Theoretical Physics, Waterloo}

\begin{abstract}

Whether or not the Kronecker coefficients of the symmetric group count some set of combinatorial objects is a longstanding open question. In this work we show that a given Kronecker coefficient is proportional to the rank of a projector that can be measured efficiently using a quantum computer. In other words a Kronecker coefficient counts the dimension of the vector space spanned by the accepting witnesses of a $\QMA$ verifier, where $\QMA$ is the quantum analogue of $\NP$. This implies that approximating the Kronecker coefficients to within a given relative error is not harder than a certain natural class of \textit{quantum approximate counting} problems that captures the complexity of estimating thermal properties of quantum many-body systems. A second consequence is that deciding positivity of Kronecker coefficients is contained in $\QMA$, complementing a recent $\NP$-hardness result of Ikenmeyer, Mulmuley and Walter.  We obtain similar results for the related problem of approximating row sums of the character table of the symmetric group. Finally, we discuss an efficient quantum algorithm that approximates normalized Kronecker coefficients to inverse-polynomial additive error. 
\end{abstract}
\maketitle

In this work we revisit a well-known mystery in combinatorics--- what do the Kronecker coefficients count?---from the perspective of quantum complexity theory.

The Kronecker coefficients are the analogues of the Clebsch-Gordan coefficients for the symmetric group. In particular, let $n\geq 2$ and let $\rho_\mu$ and $\rho_\nu$ be irreducible representations of the symmetric group labelled by partitions $\mu, \nu \vdash n$ respectively \footnote{We shall assume the reader is familiar with basic definitions concerning the representation theory of the symmetric group.}. The Kronecker coefficient $g_{ \mu \nu \lambda}$ is the multiplicity of the irreducible representation $\rho_\lambda$ in the tensor product representation $\rho_\mu \otimes \rho_\nu$:
\begin{align}
    \rho_\mu \otimes \rho_\nu \simeq \bigoplus_\lambda g_{\mu \nu \lambda } \rho_\lambda.
\label{eq:tensorrep}
\end{align}
where the $\simeq$ stands for isomorphism. 
Although we will not need it here, we note that a related notion of \textit{stretched} Kronecker coefficients \footnote{A stretched Kronecker coefficient is a family of Kronecker coefficients, where the number of boxes in every partition on every input can be multiplied by an arbitrary positive integer} plays an important role in the study of the quantum marginal problem~\cite{christandl05, christandl2007nonzero, klyachko04}.

The definition Eq.~\eqref{eq:tensorrep} of the Kronecker coefficients as multiplicities ensures that they are nonnegative integers. But it is a longstanding open problem in algebraic combinatorics---Problem~10 on Stanley's list \cite{stanley99positivityproblems}---to find a combinatorial formula for them. That is, \textit{do the Kronecker coefficients count some natural set of combinatorial objects?} Let us interpret ``natural" as any set of objects where the description size is polynomial in $n$ and such that membership in the set can be verified efficiently using a classical computer. In this way we are led to the complexity-theoretic question \cite{burgisser2008complexity,pak14,ikenmeyer2017vanishing, pak2020breaking, Ikenmeyer22b, paksurvey}: does $g_{\mu\nu\lambda}$ admit a $\#\P$ formula, i.e., does it count the number of accepting witnesses of an $\NP$ verifier? 

It will be convenient to define the problem of computing the Kronecker coefficients as follows. Here, and in other computational problems defined below, the input partitions $\mu,\nu,\lambda\in S_n$ are given in unary so that the input size is $O(n)$ \footnote{Some works in this area consider problems where the parts of each partition are specified in binary (i.e., Ref.~\cite{burgisser2008complexity}).}.
\begin{problem}[ExactKron]
Given $\mu,\nu,\lambda \vdash n$, compute $g_{\mu\nu\lambda}$. 
\end{problem}

There is little doubt that ExactKron is an extremely hard computational problem. It is known that any problem in $\#\P$ can be reduced in polynomial time to ExactKron  \cite{ikenmeyer2017vanishing}\footnote{See Remark 14 in Ref.~\cite{pak2020breaking}},  and that ExactKron is contained in the class $\GapP$ \cite{burgisser2008complexity} of functions that can be expressed as the difference $f-g$ of two functions $f$ and $g$ in $\#\P$ \cite{gapp}. In this sense, the Kronecker coefficients are as hard as $\#\P$, but also not much harder. Pinning down their complexity is therefore entirely concerned with the sliver of daylight between $\#\P$ and $\GapP$. Why should we care? 

 Firstly, as discussed above, if ExactKron is contained in $\#\P$ then this would resolve a significant open problem in combinatorics by showing that Kronecker coefficients count the elements of an efficiently verifiable set.  A second important (though less direct) consequence has to do with the complexity of approximating Kronecker coefficients to within a given relative error.  

\begin{problem}[ApproxKron]
Given $\mu,\nu,\lambda \vdash n$ and $\epsilon=\Omega(1/\poly(n))$, compute an estimate $\tilde{g}$, such that:
\[
(1-\epsilon)g_{\mu\nu\lambda} \leq \tilde{g}\leq (1+\epsilon) g_{\mu\nu\lambda}.
\]
\end{problem}

This problem is at least as hard as deciding if $g_{\mu \nu\lambda} >0$, which has been recently shown to be $\NP$-hard \cite{ikenmeyer2017vanishing}. It is also no harder than ExactKron. In order to discuss its complexity in more detail, we now review some facts about the relationship between exact and approximate counting problems more generally. 

Some problems in $\GapP$ are just as hard to approximate as they are to compute exactly (with respect to polynomial-time reductions). This is the case, for example, for the problems of computing output probabilities of polynomial-size quantum circuits \cite{goldberg2017complexity}, or computing the squared permanent of a real matrix \cite{aaronson2011computational}, or computing the square of a given irreducible character of the symmetric group \cite{Ikenmeyer22}. Approximating these quantities to within a given relative error is $\#\P$-hard. 
\begin{table*}[t!]
\centering
\begin{tabular}{ |l|c|c|} 
 \hline
& Problem in...& Approximation problem upper bound...\\
\hline
Classical counting & \#\P & $\FBPP^{\NP}$ \\
\hline 
 Quantum counting  & $\#\BQP$ & $\QAPC$\\ 
\hline
Gap counting & \GapP & $\#\P$ \\ 
 \hline
\end{tabular}
\caption{\label{3types} Three classes of counting problems that are polynomial-time reducible to one another. The associated approximation problems are very unlikely to be polynomial-time equivalent. Classical counting problems in $\#\P$ can be approximated within the third level of the polynomial hierarchy via Stockmeyer's approximate counting algorithm. In contrast, some problems in $\GapP$ can be $\#\P$-hard to approximate. Quantum approximate counting problems lie somewhere between these two extremes.}
\end{table*}
On the other hand, Stockmeyer has shown that the hardest functions in $\#\P$ are \textit{vastly easier to approximate} assuming a standard conjecture in complexity theory \cite{stockmeyer1983complexity}. In particular, if a function is contained in $\#\P$, then it can be approximated to a given relative error using a polynomial-time randomized algorithm that has access to an $\NP$ oracle, i.e., in the class $\FBPP^{\NP}$. Problems in this class cannot be $\#\P$-hard unless the polynomial hierarchy collapses, which is believed to be very unlikely. 

So far we have considered two classes of counting problems---$\GapP$ and $\#\P$ (which are polynomial-time reducible to one another) as well as the approximation tasks associated with them (which are unlikely to be polynomial-time equivalent). A central message of this work and Ref.~\cite{bravyi2022quantum} is that in between these two extremes there is a rich class of approximation tasks associated with a third type of \textit{quantum} counting problem.

 By analogy with the definition of $\#\P$, Refs.~\cite{shi2009note,brown2011computational} defined the class $\#\BQP$ of functions that can be thought of informally as counting the number of accepting witnesses of a $\QMA$ verifier. Here $\QMA$ is a quantum analogue of $\NP$--- a class of problems for which yes instances can be verified with high probability using a polynomial-sized quantum computation. A quantum verification circuit $C$ takes as input a quantum state $\ket{\psi}$ as well as some number $a$ of ancilla qubits in the $|0\rangle$ state. At the end of the circuit the first qubit is measured and the circuit accepts the input state $|\psi\rangle$ if and only if the outcome $1$ is observed. This occurs with probability $p(\psi)=\||(|1\rangle\langle1|\otimes I) C|\psi\rangle|0^{a}\rangle\|^2.$ 

For any quantum verification circuit we define the POVM element corresponding to the accept outcome
\begin{equation}
A\equiv(I\otimes\langle 0^{a}|)C^{\dagger} (|1\rangle\langle1|\otimes I) C(I\otimes |0^{a}\rangle), 
\label{eq:POVM}
\end{equation}
and for $q\in [0,1]$ we define $N_q$ to be the number of eigenvalues of $A$ that are at least $q$ (informally, the number of witnesses accepted with probability at least $q$). $\#\BQP$ is the class of problems that can be expressed as: given a quantum verification circuit with an $n$-qubit input register and total size $\mathrm{poly}(n)$, along with two thresholds $0<b<a\leq 1$ satisfying $a-b=\Omega(1/\mathrm{poly}(n))$, compute an estimate $N$ such that $N_a\leq N\leq N_b$. If we are assured that there are no ``mediocre witnesses" ~\cite{brown2011computational}, i.e., eigenvalues of $A$ in the interval $[b,a)$, then $N=N_a$ is the number of witnesses accepted with probability at least the given threshold acceptance probability $a$. We note that this definition of $\#\BQP$, allowing for a ``grace interval'' of miscounted witnesses \cite{brown2011computational}, was suggested in Ref.~\cite{brown2011computational} but passed over in favor of a slightly different definition. Though our results do not rest on this technical choice, we find the above to be the most natural definition for our purposes (see the Supplementary Material for more details).

The corresponding class $\QAPC$ of \textit{quantum approximate counting problems} \cite{bravyi2022quantum} is then the approximation version of $\#\BQP$. 
Here the goal is to approximate the number of accepting witnesses of a $\QMA$ verifier
within a small relative error.
That is, given an $n$-qubit quantum verification circuit of size $\mathrm{poly}(n)$, an error parameter $\delta=\mathrm{poly}(1/n)$ and two thresholds $0<b<a\leq 1$ and $a-b=\Omega(1/\mathrm{poly}(n))$, compute an estimate $N$ such that $(1-\delta) N_a\leq N\leq N_b(1+\delta)$.  

Quantum approximate counting was studied in Ref.~\cite{bravyi2022quantum} where it was shown that a wide class of problems that arise naturally in quantum many-body physics are equivalent to $\QAPC$ under polynomial-time reductions \footnote{In this paper we have changed the name to $\QAPC$ from $\mathrm{QAC}$ as in Ref.~\cite{bravyi2022quantum} to avoid confusion with other quantum complexity classes}. These include the problem of approximating the partition function of a quantum many-body system described by a local Hamiltonian (to within a given relative error), as well as the problem of computing expected values of local observables in the thermal (Gibbs) state at a given temperature (to within a prescribed additive error). Despite the centrality of these problems in physics, the complexity of quantum approximate counting is largely an open question. The main challenge is that there does not appear to be a natural quantum generalization of Stockmeyer's approximate counting method (a similar obstacle was encountered in attempting to generalize the Valiant-Vazirani theorem \cite{aharonov2022pursuit, jain2012power}). 
All we can say is that quantum approximate counting is at least $\QMA$-hard~\footnote{Note 
that  approximating the number of accepting witnesses of a $\QMA$ verifier 
with a small relative error allows one to check whether the verifier accepts at least one witness.}
and at most $\#\P$-hard (because quantum exact counting is 
as hard as classical exact counting~\cite{brown2011computational}). Nevertheless we conjecture that quantum approximate counting is not $\#\P$-hard---i.e., that it defines an intermediate class that lies somewhere between classical approximate counting and $\#\P$. Tab.~\ref{3types} summarizes the three types of counting problems discussed above, and their approximation versions.

In this work we describe an efficient quantum circuit that measures a projector with rank $d_{\mu}d_{\nu}d_{\lambda} g_{\mu\nu\lambda}$. Since the dimension $d_{\omega}$ of any irreducible representation $\omega$ of $S_n$ can be computed efficiently via the hook-length formula, this shows that $g_{\mu\nu\lambda}$ is equal to an efficiently computable dimensional factor $(d_{\mu}d_{\nu}d_{\lambda})^{-1}$ times a $\#\BQP$ function. Therefore the problem of approximating Kronecker coefficients to within a given relative error (ApproxKron) is not harder than $\QAPC$. This furnishes one of the first natural computational problems that is polynomial-time reducible to quantum approximate counting (and not known to be easier) but whose definition does not involve quantum many-body systems. The only other problem of this kind that we are aware of involves approximating Betti numbers of a cochain complex \cite{cade2021complexity}\footnote{In particular, the approximation version of the problem described in Theorem 5 of Reference \cite{cade2021complexity}}. Our results make use of Beals' efficient quantum Fourier transform over the symmetric group \cite{beals1997quantum} and, in particular, its use in the \textit{generalized phase estimation algorithm} \cite{harrow05} which we now review.

Let $\rho$ be a representation of the symmetric group $S_n$ such that for every permutation $\sigma\in S_n$,  $\rho(\sigma)$ is a unitary operator acting on a Hilbert space $\mathcal{H}$.  For any irreducible representation of the symmetric group labeled by a partition $\lambda \vdash n$, we define an operator acting on $\mathcal{H}$ as
\begin{equation}
\Pi_{\lambda}=\frac{d_{\lambda}}{n!}\sum_{\alpha\in S_n} \chi_{\lambda}(\alpha) \rho(\alpha).
\label{eq:pilambda}
\end{equation}
Here $\chi_{\lambda}:S_n\rightarrow \mathbb{C}$ is the character of $S_n$ corresponding to the irreducible representation labeled by $\lambda$ (it is a nontrivial fact that $\chi_{\lambda}$ is in fact integer-valued \cite{james1984}), and $d_\lambda$ is its dimension. Also note that the operator Eq.~\eqref{eq:pilambda} depends on the representation $\rho$, although this is suppressed in our notation.

The operators $M_{\rho}\equiv\{\Pi_{\lambda}: \lambda \vdash n\}$ define a projective measurement. That is, as shown in the Supplementary Material, the $\Pi_{\lambda}$ are mutually orthogonal Hermitian projectors that sum to the identity. 

In fact, this measurement can almost always (i.e., for most representations $\rho$) be implemented efficiently on a quantum computer. First consider the simple case when $\mathcal{H}=\mathbb{C}S_n$ and $\rho$ is the left-regular representation acting as $\rho(\alpha)|\pi\rangle=|\alpha\pi\rangle$ for any permutations $\alpha,\pi\in S_n$. Let us denote the corresponding projectors Eq.~\eqref{eq:pilambda} as $\Pi^{L}_{\lambda}$ and the measurement as $M_{L}=\{\Pi_{\lambda}^{L}\}$. This measurement---also known as weak Fourier sampling over the symmetric group--- can be implemented using a simple circuit depicted in Fig.~\ref{fig:weakfourier}. Here $\mathrm{QFT}_n$ is the unitary corresponding to Beals' quantum Fourier transform over the symmetric group \cite{beals1997quantum}. It is defined in terms of Young's orthogonal representation of $S_n$ which describes an irreducible matrix representation of $S_n$, denoted $\rho_{\omega}$, for each $\omega\vdash n$, with certain desirable properties. For any function $f:S_n \rightarrow \mathbb{C}$ we have
 \begin{equation}
 \mathrm{QFT}_n\sum_{\sigma\in S_n} f(\sigma)|\sigma\rangle=\sum_{\omega \vdash n} \sqrt{\frac{d_{\omega}}{n!}} \sum_{i,j=1}^{d_{\omega}} (\hat{f}(\omega))_{ij}  |\omega, i,j\rangle, 
\label{eq:ft1}
\end{equation}
where
\begin{equation}
 \hat{f}(\omega)=\sum_{\alpha\in S_n} f(\alpha)\rho_\omega(\alpha).
\label{eq:ft2}
 \end{equation}
(Note that basis vectors of $\mathcal{H}$ can be labeled by triples from the set $T=\{(\omega, i,j): \omega\vdash n, 1\leq i,j\leq d_{\omega}\}$, since $|T|=\sum_{\omega\vdash n} d_{\omega}^2 =|S_n|$).

For other unitary representations $\rho$ of $S_n$, it is possible to implement the measurement $M_{\rho}$ using the \textit{generalized phase estimation} \cite{harrow05} circuit shown in Fig.~\ref{fig:gpe}, which uses weak Fourier sampling as a subroutine. It further requires that we can perform the $\mathrm{C}-\rho$ gate and its inverse, where $(\mathrm{C-}\rho)|\alpha\rangle|\psi\rangle=|\alpha\rangle \rho(\alpha)|\psi\rangle$ for any permutation $\alpha\in S_n$ and $\psi\in \mathcal{H}$.

\begin{figure*}[t!]
\begin{flushleft}
\begin{subfigure}{0.35\textwidth}
\begin{tikzpicture}
\draw (0,0)--(1,0);
\draw(1, -1) rectangle (2,1);
\node at (1.5,0) {$\mathrm{QFT}_n$};
\draw (2,0.5)--(4,0.5);
\draw (2,0)--(4,0);
\draw (2,-0.5)--(4,-0.5);

\node at (2.2,0.65) {$\omega$};
\node at (2.2,0.15){$i$};
\node at (2.2, -0.3) {$j$};
\draw (4, -1) rectangle (5,1);
\node at (4.5,0) {$\mathrm{QFT}_n^{\dagger}$};
\draw (5,0)--(6,0);
\node at (1,1.5) {$|0\rangle$};
\node at (-0.5,0) {$|\psi\rangle$};
\draw (1.2,1.5)--(4,1.5);
\draw (4, 1.3) rectangle (5,1.8);
\draw[black,fill=black] (3,0.5) circle (0.5ex);
\draw (3,1.64)--(3,0.5);
\draw[black] (3,1.5) circle (0.13cm);
\draw[black] (4.2,1.5) arc (150:30:0.30cm);
\draw [-stealth] (4.4,1.5)--(4.6,1.75);
\draw (5,1.48)--(6,1.48);
\draw (5,1.55)--(6,1.55);
\node at (6.2, 1.5) {$\omega$};
\node at (6.5,0){$\Pi_{\omega}^L|\psi\rangle$};
\draw[thick,dotted] (0.5,-1.5) rectangle (5.5,2);
\end{tikzpicture}
\end{subfigure}
\hspace{3cm}
\begin{subfigure}{0.35\textwidth}
\begin{tikzpicture}
\draw (-0.5,0)--(1,0);
\draw (1,-0.5) rectangle (2,0.5);

\draw[black,fill=black] (1.5,1) circle (0.5ex);
\draw (1.5,1)--(1.5,0.5);
\draw (2,0)--(4,0);
\draw[black,fill=black] (4.5,1) circle (0.5ex);
\draw(4.5,1)--(4.5,0.5);
\draw (4,-0.5) rectangle (5,0.5);
\draw (5,0)--(5.5,0);
\draw (0.9,1)--(2.6,1);
\draw (2.6,.6) rectangle (3.4,1.75);
\draw (3.4, 1.43)--(5.5,1.43);
\draw (3.4, 1.5)--(5.5,1.5);
\draw (3.4,1)--(5.5,1);
\node at (3,1.25) {$M_L$};
\node at (1.5,0) {$\rho^{\dagger}$};
\node at (4.5,0) {$\rho$};
\node at (-1,0){$|\phi\rangle$};
\node at (6,0){$\Pi_{\omega}|\phi\rangle$};
\node at (6,1.48){$\omega$};

\draw (0, 0.6) rectangle (0.9,1.75);
\node at (0.45,1.25) {$\mathrm{QFT}_n^{\dagger}$};
\draw (-0.5,1.1)--(0,1.1);
\draw (-0.5,1.4)--(0,1.4);
\draw (-0.5,0.8)--(0,0.8);
\node at (-1,1.4) {$|\tau\rangle$};
\node at (-1,1.1) {$|1\rangle$};
\node at (-1,0.8) {$|1\rangle$};
\end{tikzpicture}
\end{subfigure}
\caption{
(a) Weak Fourier sampling circuit that performs the projective measurement $M_{L}=\{\Pi_{\omega}^{L}: \omega\vdash n\}$.\label{fig:weakfourier}
(b) Generalized phase estimation algorithm\label{fig:gpe}. Here $\tau$ is the trivial representation of $S_n$, such that $\mathrm{QFT}_n^{\dagger}|\tau,1,1\rangle=\frac{1}{\sqrt{n!}}\sum_{\alpha\in S_n}|\alpha\rangle$.
}
\end{flushleft}
\end{figure*}
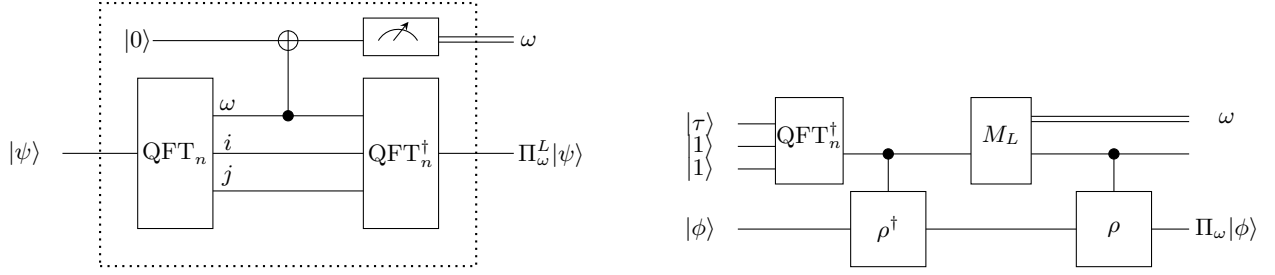

We are now ready to describe our main results.  Again let us consider the left-regular representation of $S_n$.  For any three partitions $\mu,\nu,\lambda$ of $n$, we may define the projectors $\Pi^L_{\mu}, \Pi^L_{\nu}, \Pi^L_{\lambda}$ using Eq.~\eqref{eq:pilambda} where $\rho$ is the left regular representation of $S_n$.These operators act on the Hilbert space $\mathcal{H}=\mathbb{C}S_n$.

Next consider $\mathcal{H}^{\otimes 3}$ and let 
\begin{equation}
Q\equiv\frac{1}{n!} \sum_{\sigma\in S_n} \sigma \otimes \sigma \otimes \sigma.
\label{eq:Q}
\end{equation}
Then it is easy to verify that $Q$ is a hermitian projector, i.e.,  $Q^2=Q=Q^{\dagger}$. Moreover, $Q$ is a projector of the form Eq.~\eqref{eq:pilambda} where $\rho(\sigma)=\sigma\otimes \sigma\otimes \sigma$ is the threefold tensor product representation and $\lambda$ is the trivial (one-dimensional) irreducible representation $\tau$ such that $\rho_{\tau}(\sigma)=\chi_{\tau}(\sigma)=1$ for all $\sigma\in S_n$.

Using the fact that $[\sigma, \Pi^{L}_{\alpha}]=0$ for any $\sigma\in S_n$ and $\alpha \vdash n$, we get $
[Q,\Pi^{L}_{\mu}\otimes \Pi^{L}_{\nu}\otimes \Pi^{L}_{\lambda}]=0$. Let us then define
\[
P_{\mu\nu\lambda}\equiv Q\big(\Pi^{L}_{\mu}\otimes \Pi^{L}_{\nu}\otimes \Pi^{L}_{\lambda}\big)=\big(\Pi^{L}_{\mu}\otimes \Pi^{L}_{\nu}\otimes \Pi^{L}_{\lambda}\big)Q
\]
as the projector onto the simultaneous $+1$ eigenspace of $Q$ and $\Pi^{L}_{\mu}\otimes \Pi^{L}_{\nu}\otimes \Pi^{L}_{\lambda}$. Then 
\begin{lemma}$ g_{\mu\nu\lambda}=\frac{1}{d_{\mu}d_\nu d_\lambda}\mathrm{Tr}(P_{\mu\nu\lambda})$.
\label{lem:numkron}
\end{lemma}
\begin{proof}
 Taking the trace in Eq.~\eqref{eq:tensorrep} we see that $\chi_\mu(\sigma) \chi_\nu(\sigma) = \sum_\lambda g_{\mu  \nu \lambda } \chi_\lambda(\sigma)$ for all $\sigma \in S_n$. Using orthogonality of the irreducible characters, we then obtain the well-known formula 
\begin{equation}
\label{eq:g}
g_{\mu \nu \lambda} = \frac{1}{n!}\sum_{\sigma \in S_n} \chi_\mu(\sigma) \chi_\nu(\sigma) \chi_\lambda(\sigma).
\end{equation}
We can evaluate $\mathrm{Tr}(P_{\mu\nu\lambda})$ using Eqs.~(\ref{eq:pilambda},\ref{eq:Q}) and the fact that in the left-regular representation $\mathrm{Tr}(\pi)$ is zero unless $\pi$ is the identity permutation:
\begin{align*}
&\frac{1}{d_{\mu}d_{\nu}d_{\lambda}}\mathrm{Tr}(P_{\mu\nu\lambda})=\frac{1}{d_{\mu}d_{\nu}d_{\lambda}}\mathrm{Tr}\left(Q\big(\Pi^{L}_{\mu}\otimes \Pi^{L}_{\nu}\otimes \Pi^{L}_{\lambda}\big)\right)\\
&=\frac{1}{n!}\sum_{\substack{\sigma,\alpha \\ \beta,\gamma}} \delta_{\sigma,\alpha^{-1}}\delta_{\sigma,\beta^{-1}} \delta_{\sigma, \gamma^{-1}}\chi_{\mu}(\alpha)\chi_\nu(\beta)\chi_\lambda(\gamma),
\end{align*}
which coincides with Eq.~\eqref{eq:g}.
\end{proof}

\begin{theorem}
For any partitions  $\mu,\nu,\lambda\vdash n$ there exists a quantum circuit of size $\mathrm{poly}(n)$ that implements a projective measurement $\{P_{\mu\nu\lambda}, I-P_{\mu\nu\lambda}\}$. 
\label{thm:main}
\end{theorem}
\begin{proof}
The circuit will have an input register $\mathbb{C}S_n\otimes \mathbb{C}S_n\otimes \mathbb{C}S_n$.
First we use weak Fourier sampling to perform the measurement $M_L$ on each of the three registers and reject unless the outcomes are $\mu,\nu,\lambda$ respectively. Then we use generalized phase estimation with the threefold tensor product representation $\rho(\sigma)=\sigma\otimes \sigma\otimes \sigma$  and reject unless the outcome is the trivial representation $\tau$. If the circuit does not reject in either step, then it accepts. The corresonding POVM element is the product of the POVM elements for acceptance in each step individually, that is $ Q\big(\Pi^{L}_{\mu}\otimes \Pi^{L}_{\nu}\otimes \Pi^{L}_{\lambda})=P_{\mu\nu\lambda}$.
\end{proof}

\begin{corol}
The function $(\mu,\nu,\lambda)\mapsto d_{\mu}d_{\nu}d_{\lambda} g_{\mu\nu\lambda}$ is in $\#\BQP$, and ApproxKron is polynomial-time reducible to $\QAPC$.
\label{corol:main}
\end{corol}
\begin{proof}
From Theorem \ref{thm:main} there is a polynomial size verification circuit that measures the projector $P_{\mu\nu\lambda}$. Since it is a projector, its eigenvalues are either zero or 1. By Lemma \ref{lem:numkron} the number of $+1$ eigenvalues is equal to $d_{\mu}d_{\nu}d_{\lambda} g_{\mu\nu\lambda}$. By definition it is a $\#\BQP$ problem to exactly count the number of $+1$ eigenvalues, and a $\QAPC$ problem to approximate the number of $+1$ eigenvalues to within a given relative error $\delta$ (for example, we can set $a=1$ and $b=1/2$ in the definitions of $\#\BQP$ and $\QAPC$).  Since the dimensional factor $d_{\mu}d_{\nu}d_{\lambda}$ is efficiently computable this shows that ApproxKron is polynomial time reducible to $\QAPC$.
\end{proof}

The efficient quantum circuit that measures the projector $P_{\mu\nu\lambda}$ (Theorem \ref{thm:main}) directly gives a $\QMA$ protocol for the problem of deciding positivity of Kronecker coefficients. Indeed, the projector $P_{\mu\nu\lambda}$ has nonzero rank if and only if $g_{\mu\nu\lambda}>0$. If $g_{\mu\nu\lambda}>0$ there exists a quantum state that results in the measurement outcome $+1$ while if $g_{\mu\nu\lambda}=0$ then all states give measurement outcome $0$. Note that this $\QMA$ protocol is deterministic in the sense that it has perfect completeness and soundness .
\begin{corol}
The problem of deciding positivity of Kronecker coefficients---given $\mu,\nu, \lambda\vdash n$, decide if $g_{\mu\nu\lambda}>0$--- is contained in $\QMA$.
\label{corol:qma}
\end{corol}

This upper bound complements the recent result of Ikenmeyer, Mulmuley and Walter who established that deciding positivity of Kronecker coefficients is $\NP$-hard \cite{ikenmeyer2017vanishing}.

Corollaries \ref{corol:main}, \ref{corol:qma} describe complexity-theoretic consequences of the fact that a Kronecker coefficient is proportional to the dimension of a subspace for which membership can be efficiently verified on a quantum computer. Similar arguments shed light on another positivity problem on Stanley's list \cite{stanley99positivityproblems}, which concerns the row sums of the character table of the symmetric group:
\begin{equation}
R_{\lambda}\equiv \sum_{\mu \vdash n} \chi_{\lambda}(\mu).
\label{eq:rowsum}
\end{equation}
Here $\chi_{\lambda}(\mu)$ is the value of the irreducible character $\chi_{\lambda}$ evaluated at any permutation with cycle structure described by $\mu$.

Just like the Kronecker coefficients, the row sums $R_{\lambda}$ are nonnegative integers. This follows from the fact that $R_{\lambda}$ is equal to  the multiplicity of an irreducible representation $\rho_\lambda$ in the decomposition of the so-called \emph{ conjugation representation} $\rho_c$ which acts on the Hilbert space $\mathcal{H}=\mathbb{C}S_n$ as $\rho_c(\pi)|\sigma\rangle = |\pi \sigma \pi^{-1}\rangle$. In particular, we have the decomposition $\rho_c\simeq \bigoplus_\lambda R_{\lambda}\rho_{\lambda}$, see the Supplementary Material for details. The generalized phase estimation circuit Fig.~\ref{fig:gpe}(b) can be used to efficiently measure the projector $\Pi^c_{\lambda}$ given by Eq.~\eqref{eq:pilambda} with $\rho\rightarrow\rho_c$. In this way we obtain:

\begin{theorem}
For any $\lambda\vdash n$ there exists a quantum circuit of size $\mathrm{poly}(n)$ that implements a projective measurement $\{\Pi^{c}_{\lambda}, I-\Pi^{c}_{\lambda}\}$, where  $d_{\lambda} R_\lambda= \mathrm{Tr}(\Pi^{c}_{\lambda})$. 
\label{thm:rowsum}
\end{theorem}
As a consequence, the problem of computing $d_{\lambda} R_{\lambda}$ is in $\#\BQP$, the problem of approximating $R_{\lambda}$ to within a given relative error is efficiently reducible to $\QAPC$, and the problem of deciding positivity of $R_{\lambda}$ is contained in $\QMA$.

Our discussion thus far has centred around problems that are thought to be intractable for quantum or classical computers. It is natural to ask whether there are easier versions of these problems that may admit efficient quantum or classical algorithms. For example--- although it is very unlikely that there exists an efficient algorithm to approximate the irreducible characters of the symmetric group to with a given \textit{relative} error \cite{Ikenmeyer22}---  Jordan has shown that the normalized irreducible characters $\chi_{\lambda}(\pi)/d_{\lambda}$ can be efficiently approximated to within a given additive error using a classical randomized algorithm \cite{Jordan08}. Similarly, a quantum algorithm for approximating normalized Kronecker coefficients can be inferred from Ref.~\cite{Moore07}.
\begin{lemma} \cite{Moore07}
There is a quantum circuit of size $poly(n)$ that samples from the probability distribution $p(\lambda)=  \frac{d_\lambda g_{\mu\nu\lambda}}{d_\mu d_\nu}$.
\label{lem:samplekron}
\end{lemma}
We describe the algorithm in the Supplementary Material. (To see that  $p(\lambda)$ is a normalized  probability distribution, apply Eq.~(\ref{eq:tensorrep}) to the identity permutation and take the trace.)
Given a partition $\lambda$, one can  approximate
the probability $p(\lambda)$ 
 with an additive error
$\epsilon$ by generating $O(1/\epsilon^2)$ samples from $p$ and computing
the fraction of samples equal to $\lambda$. 
This gives an approximation of $g_{\mu \nu \lambda}$ with an additive error
$\delta = \epsilon d_\mu d_\nu /d_\lambda$. 
This approximation can be efficiently obtained
on a quantum computer, as long as $\epsilon \ge \mathrm{poly}(1/n)$.
Note that this approximation
is good enough to solve
ExactKron  if $\min{(d_\mu,d_\nu,d_\lambda)}\le \mathrm{poly}(n)$.
Indeed, assuming wlog that $d_\mu\le d_\nu\le d_\lambda$, one gets
$\delta \le \epsilon d_\mu \le \epsilon \cdot \mathrm{poly}(n)<1/2$ for $\epsilon=  \mathrm{poly}(1/n)$. We do not know if such an approximation can be efficiently obtained classically.
Polynomial-time
classical algorithms for ExactKron are only known in the special case when all partitions $\mu,\nu,\lambda$ have constant (or nearly constant) number of parts~\cite{pak14}.

In summary, we have upper bounded the computational hardness of counting problems related to the representation theory of the symmetric group. We have shown that two mysterious quantities in this arena---the Kronecker coefficients and the row sums of the symmetric group--- are proportional to  $\#\BQP$ functions and can be approximated via quantum approximate counting. We have also described an efficient quantum algorithm for estimating the normalized Kronecker coefficient $g_{\mu\nu\lambda} d_{\lambda}/(d_{\mu}d_{\nu})$ to within a given additive error. 
We note that Kronecker coefficient squared $g_{\mu\nu\lambda}^2$ can be expressed as an eigenvalue degeneracy for a certain Hamiltonian describing dynamics of bipartite ribbon graphs~\cite{geloun2020quantum}.
This provides an alternative quantum mechanical interpretation of the Kronecker coefficients.
It remains a challenging open question to understand the computational hardness of quantum approximate counting relative to other known complexity classes.

\paragraph{Acknowledgments}
We thank Matthias Christandl, Robin Kothari, Hari Krovi, Greg Kuperberg, Igor Pak, and Michael Walter for helpful discussions. SB, AC, DG, and VH are supported in part by the Army Research Office under Grant Number W911NF-20-1-0014. DG is a CIFAR fellow in the quantum information science program, and is also supported in part by IBM Research.  GZ is supported by the U.S. Department of Energy, Office of Science, National Quantum Information Science Research Centers, Co-design Center for Quantum Advantage (C2QA) under contract number DE-SC0012704. This research was supported in part by Perimeter Institute for Theoretical Physics. Research at Perimeter Institute is supported by the Government of Canada through the Department of Innovation, Science and Economic Development and by the Province of Ontario through Ministry of Colleges and Universities.

\paragraph{Note added:} 
At the final stages of completing this work we became aware that other groups have independently obtained similar results \footnote{Greg Kuperberg, Igor Pak, and Greta Panova, personal communication.}\footnote{Matthias Christandl, Aram Harrow, and Michael Walter. Presented by M. Walter at the Algorithms and Complexity in Algebraic Geometry Reunion-Workshop, Simons Institute, Berkeley, December 2015}.

\bibliographystyle{plain}
\bibliography{refs}

\onecolumngrid
\appendix

\section{Properties of the operators $\{\Pi_{\mu}\}$} 
Here we show that the operators in Eq.~\eqref{eq:pilambda} are mutually orthogonal projectors that add up to identity. To this end we use the \emph{grand orthogonality theorem} (a consequence of Schur's lemma) of representation theory.
\begin{theorem}[Grand orthogonality theorem] (See \cite[Sec. 2.2.]{serre1977linear})
    Given irreducible representations $\rho_\mu$ and $\rho_\nu$, we have that: 
    \begin{align}
        \sum_{g \in S_n} \rho_\mu(g)^\dag_{ij} \rho_\nu(g)_{k\ell} &= \frac{n!}{d_\mu} \delta_{ik} \delta_{j \ell} \delta_{\mu \nu}. 
    \end{align} 
\label{thm:got}
\end{theorem}
Using Theorem \ref{thm:got} we obtain the following identity for any fixed $h \in S_n$.
\begin{align*}
    \sum_{g \in S_n} \chi_\nu(g) \chi_\mu(g^{-1} h) &= \sum_{g \in S_n} \sum_{i,j} \rho_\nu(g)_{ii} \rho_\mu(g^{-1}h)_{jj} \\ &= \sum_{g \in S_n} \sum_{ijk}\rho_\nu(g)_{ii} \rho_\mu(g^{-1})_{jk} \rho_\mu(h)_{kj} \\ &=  \sum_{ijk}  \left( \sum_{g \in S_n} \rho_\nu(g)_{ii} \rho_\mu(g)^\dag_{jk} \right) \rho_\mu(h)_{kj} \\ 
    &= \delta_{\mu \nu} \sum_{ijk} \delta_{ij} \delta_{ik} \frac{n!}{d_\nu} \rho_\mu(h)_{kj} = \delta_{\mu \nu} \sum_i \frac{n!}{d_\nu} \rho_\mu(h)_{ii} = \delta_{\mu \nu} \frac{n!}{d_\nu} \chi_\mu(h).
    \end{align*}
It follows that $\{\Pi_{\mu}\}$ are mutually orthogonal projectors. 
\begin{align}
    \Pi_\nu \Pi_\mu &= \frac{d_\mu d_\nu}{(n!)^2} \sum_{gh} \chi_\nu(g) \chi_\mu(h) \rho(gh) = \frac{d_\mu d_\nu}{(n!)^2} \sum_{g \ell} \chi_\nu(g) \chi_\mu(g^{-1} \ell) \rho(\ell) = \delta_{\mu \nu} \frac{d_\mu}{n!} \sum_\ell \chi_\mu(\ell) \rho(\ell) = \delta_{\mu \nu} \Pi_\mu.
\end{align}
Using the fact that $d_\mu = \chi_\mu(e)$, where $e \in S_n$ is the identity, we obtain
\begin{align}
    \sum_{\mu \vdash n} \Pi_\mu &= \sum_{g \in S_n} \frac{1}{n!} \rho(g) \sum_{\mu \vdash n} d_\mu \chi_\mu(g) = \sum_{g \in S_n} \frac{1}{n!} \rho(g) \sum_{\mu \vdash n} \chi_\mu(e) \chi_\mu(g).
\label{eq:sumpi}
\end{align}

The column-orthogonality of characters (stated below) implies that
\begin{align}
    \sum_\mu \chi_\mu(e) \chi_\mu(g) &= n! \delta_{eg},
\end{align}
and substituting this into Eq.~\eqref{eq:sumpi} gives $\sum_\mu \Pi_\mu = I $.

\begin{lemma}[Column-orthogonality of characters] (See \cite[Thm 16.4]{alperin2012groups})
For any $h,g\in S_n$ we have
    \begin{align}
    \sum_{\mu \vdash n} \chi_\mu(h) \chi_\mu(g) &= \begin{cases}
    |Z(g)|  &\text{ if $h$ and $g$ are conjugate.}\\ 
    0 & \text{ otherwise. }
    \end{cases}
\end{align}
where $Z(g)$ is the centralizer of $g$ in $S_n$ (the set of all group elements that commute with $g$).
\end{lemma}

\section{Beals' quantum Fourier transform over the symmetric group}
For completeness, here we review Beals' efficient algorithm for the quantum Fourier transform over the symmetric group \cite{beals1997quantum} (see also \cite{Moore03} for generalizations to other nonabelian groups).

  Let $n\geq 2$ be an integer and consider the symmetric group $S_n$ of permutations on $n$ objects. Also let $T=\{\tau_1,\tau_2,\ldots, \tau_n\}$ be a transversal of the left cosets of $S_{n-1}$ in $S_n$. Here $S_{n-1}$ is regarded as the subgroup of $S_n$ that does not permute the $n$-th object and transversality means that $S_n$ is a union of cosets $\tau_1S_{n-1},\ldots,\tau_n S_{n-1}$.

It will be convenient to fix some notation and conventions concerning the function whose Fourier transform we are interested in. We shall fix a unitary matrix representation of each irreducible representation of $S_n$. In particular, we choose the so-called Young Orthogonal representation (more on this below) and write $\rho_{\lambda}(\pi)$ for the matrix corresponding to the irreducible representation labeled by $\lambda \vdash n$ evaluated at $\pi\in S_n$, and $d_{\lambda}$ for its dimension. 

For a function $f:S_n \rightarrow \mathbb{C}$ we have
 \begin{equation}
 \mathrm{QFT}_n\sum_{\sigma\in S_n} f(\sigma)|\sigma\rangle=\sum_{\omega \vdash n} \sum_{i,j \in d_\omega} (\hat{f}(\omega))_{ij}  |\omega, i,j\rangle, 
\label{eq:ft11}
\end{equation}
where
\begin{equation}
 \hat{f}(\omega)=\sqrt{\frac{d_{\omega}}{n!}} \sum_{\sigma\in S_n} f(\sigma)\rho_\omega(\sigma).
\label{eq:ft22}
 \end{equation}
For ease of notation in this section we include the normalization factor $\sqrt{\frac{d_{\omega}}{n!}}$ in Eq.~(\ref{eq:ft22}) rather than in Eq.~\eqref{eq:ft11} which differs from our convention in the main text (Cf. Eqs.~(\ref{eq:ft1},\ref{eq:ft2})). However the quantum fourier transform $\mathrm{QFT}_n$ is defined in exactly the same way.

Note that the basis vectors in the Fourier basis can be labeled by triples from the set $\Omega=\{(\omega, i,j): \omega\vdash n, 1\leq i,j\leq d_{\omega}\}$, since $|\Omega|=\sum_{\omega\vdash n} d_{\omega}^2 =|S_n|$.
Since the left cosets of $S_{n-1}$ partition $S_n$, we may write
\begin{equation}
f(g)=\sum_{j=1}^{n} F^{j}(g) \qquad \qquad F^{j}(g)\equiv \begin{cases} f(g) &, \text{if }  g\in \tau_j S_{n-1}\\ 0 &, \text{otherwise} \end{cases}.
\label{eq:defF}
\end{equation}
Also define $f_j:S_{n-1}\rightarrow \mathbb{C}$ for $1\leq j\leq n$ by 
\begin{equation}
f_j(h)\equiv f(\tau_j h) \qquad \quad h\in S_{n-1}.
\label{eq:deff}
\end{equation}

Young's orthogonal representation has a certain \textit{adapted} property that allows us to express Eq.~\eqref{eq:ft22} in terms of Fourier transforms of the functions $f_j$ for $1\leq j\leq n$ in a simple way. In particular, 
\begin{equation}
\hat{f}(\lambda)=\sum_{k=1}^{n} \rho_{\lambda}(\tau_k) \bigoplus_{\lambda^{-}\in \Phi(\lambda) } \sqrt{\frac{d_{\lambda}}{n\cdot d_{\lambda^{-}}}}\cdot \hat{f_k} (\lambda^{-}),
\label{eq:dirsum}
\end{equation}
where 
\[
\Phi(\lambda)=\{\lambda^{-}\vdash (n-1) : \lambda^{-}\leq \lambda\}
\]
is the set of partitions of $n-1$ whose diagram differs from $\lambda$ in exactly one box. Note that the Fourier transforms $\hat{f_k}$ appearing on the right hand side of Eq.~\eqref{eq:dirsum} are over the group $S_{n-1}$.

To describe Beals' QFT algorithm it will be convenient to first define some unitary transformations that are used in the construction. In the following we shall assume that we have enough qubits to encode certain classical strings as computational basis states.

 Let $W$ denote the unitary which implements a classical reversible circuit that, given a  permutation $\pi\in S_n$, computes its factorization:
\[
W|\pi\rangle=|\tau_j\rangle|h\rangle,
\]
where $j\in [n]$ and $h\in S_{n-1}$
satisfy $\pi = \tau_j h$. Note that this factorization is unique.

For $\pi \in S_n$ let $M(\pi)$ be a unitary such that
\[
M(\pi) \sum_{\lambda\vdash n}\sum_{p,q=1}^{d_{\lambda}} (A_{\lambda})_{p,q}|\lambda,p,q\rangle= \sum_{\lambda\vdash n} \sum_{p,q=1}^{d_{\lambda} }(\rho_{\lambda}(\pi)A_{\lambda})_{p,q}|\lambda,p,q\rangle
\]
Here $(A_{\lambda})_{p,q}$ are complex coefficients. Note that this can be implemented via a controlled-$\rho_{\lambda}(\pi)$ operation on the $p$ register (where the control is $\lambda$).
 The unitary $M(\pi)$ acts trivially outside the subspace spanned by basis vectors $\{|\lambda,p,q\rangle: \lambda \vdash n, p,q\in [d_{\lambda}]\}$

For each irreducible representation $\sigma$ of $S_{n-1}$ there is a set of irreps $\lambda$ of $S_n$ such that the restriction of $\rho_{\lambda}$ to $S_{n-1}$ contains $\sigma$. These correspond to the partitions $\lambda$ that differ from $\sigma$ by adding one box. Moreover, each matrix element $p,q$ of $\rho_{\sigma}$  (where $p,q\leq d_{\sigma}$) is identified with a unique matrix element $p',q'$ of $\rho_{\lambda}$ (Here $p',q'$ depend on $p,q$ as well as $\lambda$, though our notation suppresses these dependencies).

Let $U$ be a unitary such that for any $\sigma \vdash (n-1)$ and $p,q\leq d_{\sigma}$ we have
\begin{equation}
U|0\rangle|\sigma,p,q\rangle=|0\rangle\sum_{\lambda \vdash n: \sigma \leq \lambda} \sqrt{\frac{d_{\lambda}}{n\cdot d_{\sigma}}}|\lambda,p',q'\rangle 
\label{eq:U}
\end{equation}
and such that 
\begin{equation}
U|\tau_i\rangle|\sigma,p,q\rangle=|\tau_i\rangle|\sigma,p,q\rangle  \quad \qquad i\in [n].
\label{eq:pkinvariant}
\end{equation}
 It is possible to construct such a unitary because the states appearing on the RHS of Eq.~\eqref{eq:U} for different values of $\sigma,p,q$ are orthonormal. However the action of $U$ in the rest of the Hilbert space is constrained by unitarity. The following lemma describes a property of $U$ that follows from this constraint.
\begin{lemma}
Let $P$ project onto the subspace spanned by basis states than encode partitions of $n-1$, that is, 
\[
\mathrm{span}\{|w\rangle|\sigma,p,q\rangle: w\in \{0,\tau_1,\ldots, \tau_n\}, \sigma \vdash (n-1), 1\leq p,q\leq d_{\sigma}\}.
\]
Suppose $g:S_n\rightarrow \mathbb{C}$ is any function such that $g(\pi)=0$ for all $\pi\in S_{n-1}$ (regarded as the subgroup of $S_n$ that fixes $n$). Then $PU^{\dagger}|0\rangle|\hat{g}\rangle=0$.
\label{lem:orthog}
\end{lemma}
\begin{proof}
Below we write $P=\sum_{k=0}^{n}P_k$ where $P_0=\mathrm{span}\{|0\rangle|\sigma,p,q\rangle: \sigma \vdash (n-1), 1\leq p,q\leq d_{\sigma}\}$ and $P_k=\mathrm{span}\{|\tau_k\rangle|\sigma,p,q\rangle: \sigma \vdash (n-1), 1\leq p,q\leq d_{\sigma}\}$ for $k\in [n]$.

First suppose that $\hat{r}$ is the Fourier transform of some (arbitrary) function $r:S_{n-1}\rightarrow \mathbb{C}$. Define $R:S_n\rightarrow \mathbb{C}$ such that $R(\pi)=r(\pi)$ for all $\pi\in S_{n-1}$, and $R(\pi)=0$ otherwise. Then from Eq.~\eqref{eq:U} we have
\begin{equation}
|0\rangle\sum_{\sigma\vdash n-1,p,q} (\hat{r}(\sigma))_{pq}|\sigma,p,q\rangle=U^{\dagger}|0\rangle\sum_{\lambda\vdash n, p',q'} (\hat{R}(\lambda))_{p',q'}|\lambda,p',q'\rangle. 
\label{eq:rR}
\end{equation}
Now consider a function $g:S_n\rightarrow \mathbb{C}$ as in the statement of the Lemma. Since 
$g(\pi)=0$ for all $\pi\in S_{n-1}$ we have $\langle g|R\rangle=0$ and therefore $\langle \hat{g}|\hat{R}\rangle=0$ and therefore 
$U^{\dagger}|0\rangle|\hat{g}\rangle$ is orthogonal to any state of the form on the LHS of Eq.~\eqref{eq:rR}. Noting that the LHS of Eq.~\eqref{eq:rR} is an arbitrary state in the image of $P_0$, we see that 
\begin{equation}
P_0U^{\dagger}|0\rangle|\hat{g}\rangle=0.
\label{eq:p0}
\end{equation}
 Then 
\begin{align*}
PU^{\dagger}|0\rangle|\hat{g}\rangle&=(\sum_{k=1}^{n}P_k)U^{\dagger}|0\rangle|\hat{g}\rangle+P_0U^{\dagger}|0\rangle|\hat{g}\rangle\\
&=(\sum_{k=1}^{n}P_k)U^{\dagger}|0\rangle|\hat{g}\rangle\\&=(\sum_{k=1}^{n}P_k)|0\rangle|\hat{g}\rangle\\
&=0\\
\end{align*}
where in the first step we used Eq.~\eqref{eq:p0} and in the following step we used Eq.~\eqref{eq:pkinvariant}. 
\end{proof}

Our definitions of $U, f_k, F^k$ and $M(\tau_k)$ are chosen so that the following holds.
\begin{claim}
For $1\leq k\leq n$ we have $M(\tau_k)U|0\rangle|\hat{f_k}\rangle=|0\rangle|\hat{F^k}\rangle$.
\label{claim:Fkfk}
\end{claim}
\begin{proof}
Follows by combining Eqs.~(\ref{eq:dirsum}, \ref{eq:U}, \ref{eq:defF}, \ref{eq:deff}).
\end{proof}
Finally, for $k\in [n]$ let $V_k$ be the unitary that acts as
\begin{align}
V_k|0\rangle|\sigma,p,q\rangle &=|\tau_k\rangle|\sigma,p,q\rangle \\
V_k|\tau_k\rangle|\sigma,p,q\rangle &=|0\rangle|\sigma,p,q\rangle,
\end{align}
for all $\sigma\vdash n-1$ and $1\leq p,q\leq d_{\sigma}$,
and which acts as the identity on all other computational basis states.

The quantum Fourier transform over the symmetric group is described in Algorithm \ref{alg:qft}.
\begin{algorithm}[H]
\caption{Implements a unitary $\mathrm{QFT_n}$\label{alg:qft}}
\hspace*{\algorithmicindent}  \hspace{-22pt}\textbf{Input:}  A state $|f\rangle=\sum_{g\in S_n} f(g)|g\rangle$.\\
\hspace*{\algorithmicindent} \hspace{-27pt} \textbf{Output:} The state $|\hat{f}\rangle=\mathrm{QFT}_n|f\rangle$ corresponding to the Fourier transform of $f$ over $S_n$.\\

\begin{algorithmic}[1]
	\State{$|\phi\rangle\leftarrow W|\phi\rangle$} \Comment{$|\phi\rangle=\sum_{j=1}^{n} |\tau_j\rangle\sum_{h\in S_{n-1}} f_j(h)|h\rangle$}.
	\State{$|\phi\rangle \leftarrow (I\otimes \mathrm{QFT}_{n-1})|\phi\rangle$} \Comment{$|\phi\rangle=\sum_{j=1}^{n} |\tau_j\rangle \sum_{\sigma \vdash n-1} (\hat{f_j}(\sigma))_{pq} |\sigma,p,q\rangle$}
			\For{$k=1$ to $n$}
		\State{$|\phi\rangle\leftarrow M(\tau_k)^{-1}|\phi\rangle$ 
		}
			\State{$|\phi\rangle\leftarrow UV_kU^{\dagger}|\phi\rangle$}
		\State{$|\phi\rangle\leftarrow M(\tau_k)|\phi\rangle$ } 
			      \EndFor \Comment{$|\phi\rangle=|0\rangle|\hat{f}\rangle$ (see Lemma \ref{lem:forloop})} 
\State{\textbf{return} the second register $|\hat{f}\rangle$}.
		\end{algorithmic}
\end{algorithm}

The following lemma shows that the algorithm performs the quantum Fourier transform, as claimed.

\begin{lemma}
For $0\leq k\leq n$, the state after the $k$th iteration of the for loop in Algorithm \ref{alg:qft} is
\begin{equation}
|\phi_k\rangle=|0\rangle\sum_{j=1}^{k}|\hat{F^{j}}\rangle+\sum_{j=k+1}^{n} |\tau_j\rangle |\hat{f_j}\rangle.
\label{eq:phik}
\end{equation}
When $k=0$ this describes the state before the first iteration (in this case the first term is not present). For $k=n$ we have $|\phi_n\rangle=|0\rangle\sum_{j=1}^{n}|\hat{F^{j}}\rangle=|0\rangle|\hat{f}\rangle$.
\label{lem:forloop}
\end{lemma}
\begin{proof}
By induction on $k$. The base case $k=0$ corresponds to the initial state $\sum_{j=1}^{n} |\tau_j\rangle |\hat{f_j}\rangle$. 

 Now suppose the state after the $(k-1)$-th iteration is given by $|\phi_{k-1}\rangle$ as described by Eq.~\eqref{eq:phik}. Then after line 4 of Algorithm \ref{alg:qft} the state is
\begin{align}
M(\tau_k^{-1})|\phi_{k-1}\rangle&=|0\rangle\sum_{j=1}^{k-1}M(\tau_k^{-1})|\hat{F^{j}}\rangle+\sum_{j=k}^{n} |\tau_j\rangle |\hat{f_j}\rangle.
\end{align}
For each $j\leq k-1$, the state $M(\tau_k^{-1})|\hat{F^j}\rangle$ is the Fourier transform of a function 
$g(\pi)=F^j(\tau_k \pi)$ which is zero on $S_{n-1}$.
So for $j\leq k-1$, by Lemma \ref{lem:orthog} we have $P U^{\dagger}|0\rangle M(\tau_k^{-1})|\hat{F^j}\rangle =0$ where $P$ projects onto basis states that encode partitions of $n-1$. Therefore $V_k$ acts trivially on $U^{\dagger}|0\rangle M(\tau_k^{-1})|\hat{F^j}\rangle$  for all $j\leq k-1$, i.e., 
\begin{align}
V_kU^{\dagger}M(\tau_k^{-1})|\phi_{k-1}\rangle&=U^{\dagger}|0\rangle\sum_{j=1}^{k-1}M(\tau_k^{-1})|\hat{F^j}\rangle+\sum_{j=k}^{n} V_k|\tau_j\rangle |\hat{f_j}\rangle\\
&=U^{\dagger}|0\rangle\sum_{j=1}^{k-1}M(\tau_k^{-1})|\hat{F^j}\rangle+|0\rangle|\hat{f_k}\rangle+\sum_{j=k+1}^{n} |\tau_j\rangle |\hat{f_j}\rangle.
\end{align}
and, applying $M(\tau_k)U$ to the above gives
\begin{align}
|\phi_k\rangle&=|0\rangle\sum_{j=1}^{k-1}|\hat{F^j}\rangle+M(\tau_k)U|0\rangle|\hat{f_k}\rangle+\sum_{j=k+1}^{n} |\tau_j\rangle |\hat{f_j}\rangle\\
&=|0\rangle\sum_{j=1}^{k}|\hat{F^j}\rangle+\sum_{j=k+1}^{n} |\tau_j\rangle |\hat{f_j}\rangle
\end{align}
where we used Claim \ref{claim:Fkfk}. This completes the induction, and the proof.
\end{proof}

\section{ Generalized phase estimation}
In this section we will show that the circuits from Figure~\ref{fig:gpe} implement weak Fourier sampling and generalized phase estimation respectively. The following Lemma shows that the circuit in Figure \ref{fig:gpe}(a) implements weak Fourier sampling. We will make use of the inverse Fourier transform $\mathrm{QFT}_n^{\dagger}$ which acts as:
\begin{equation}
\mathrm{QFT}_n^{\dagger} \sum_{\omega \vdash n}\sum_{1\leq i,j\leq d_{\omega}} (c(\omega))_{ij} |\omega, i,j\rangle=\sum_{\sigma\in S_n}\sum_{\lambda \vdash n} \sqrt{\frac{d_\lambda}{n!}} \trace{\left(c(\lambda) \rho_{\lambda}(\sigma)^{\dagger}\right)}|\sigma\rangle.
\label{eq:invf}
\end{equation}

 \begin{lemma}
 The POVM $M_L$ can be implemented by applying $\mathrm{QFT}_n$, then performing a projective measurement of the representation label $\omega$, then appplying $\mathrm{QFT}_n^{\dagger}$.
 \label{Lemma:partition_measurement}
 \end{lemma}
 \begin{proof}
 Let a partition $\lambda \vdash n$ be given.  Let $P_{\lambda}$ denote the projector such that $P_{\lambda}|\omega,i,j\rangle=\delta_{\omega, \lambda} |\omega,i,j\rangle.$ It suffices to show that for any state $|\psi\rangle\in \mathcal{H}$ we have $\Pi^L_{\lambda}|\psi\rangle=\mathrm{QFT}_n^{\dagger} P_{\lambda} \mathrm{QFT}_n|\psi\rangle$. To this end let $|\psi\rangle=\sum_{\sigma \in S_n} f(\sigma)|\sigma\rangle$. Using Eq.~\eqref{eq:pilambda} we get
 \begin{align}
 \Pi_{\lambda}|\psi\rangle& =\frac{d_{\lambda}}{n!} \sum_{\sigma\in S_n}\sum_{\alpha\in S_n}\chi_{\lambda}(\alpha)f(\sigma)|\alpha \sigma\rangle\\
 &=\frac{d_{\lambda}}{n!}\sum_{\sigma,\beta\in S_n} \chi_{\lambda}(\beta\sigma^{-1})f(\sigma)|\beta\rangle =\frac{d_{\lambda}}{n!}\sum_{\sigma,\beta\in S_n} \chi_{\lambda}(\sigma \beta^{-1})f(\sigma)|\beta\rangle,
 \label{eq:chif1}
 \end{align}
 where in the last line we used the fact that $\chi_\lambda (g)=\chi_{\lambda}(g^{-1})$ for all $g\in S_n$  (this can be seen for example using the fact that group characters are class functions and $g$ and $g^{-1}$ are always in the same conjugacy class since their cycle structure coincides). Now using Eq.~\eqref{eq:chif1} and the fact that 
 \[
 \chi_\lambda(\sigma \beta^{-1})=\trace{\left(\rho_\lambda(\sigma \beta^{-1})\right)}=\trace{\left(\rho_\lambda(\sigma)\rho_{\lambda}( \beta)^{\dagger}\right)}
 \]
 gives
 \begin{equation}
 \Pi^L_{\lambda}|\psi\rangle=\sqrt{\frac{d_{\lambda}}{n!}}\sum_{\beta\in S_n}\trace{\left(\hat{f}(\lambda)\rho_{\lambda}(\beta)^{\dagger}\right)}|\beta\rangle\label{eq:povm1}.
 \end{equation}
 
 Below we show this is equal to $\mathrm{QFT}_n^{\dagger}P_\lambda \mathrm{QFT}_n|\psi\rangle$. We have:
 \begin{align}
 \mathrm{QFT}_n^{\dagger}P_\lambda \mathrm{QFT}_n|\psi\rangle&= \mathrm{QFT}_n^{\dagger}P_\lambda\sum_{\omega\vdash n}\sum_{i,j=1}^{d_{\omega}} (\hat{f}(\omega))_{ij}|\omega,i,j\rangle \\
 &=\mathrm{QFT}_n^{\dagger}  \sum_{i,j=1}^{d_{\lambda}} (\hat{f}(\lambda))_{ij}|\lambda,i,j\rangle =\sqrt{\frac{d_{\lambda}}{n!}}\sum_{\beta\in S_n}\trace{\left(\hat{f}(\lambda)\rho_{\lambda}(\beta)^{\dagger}\right)}|\beta\rangle,
 \end{align}
 where we used Eq.~\eqref{eq:invf} in the last line. This coincides with Eq.~\eqref{eq:povm1} and completes the proof.
 \end{proof}
Next we show that the circuit in Fig.~\ref{fig:gpe} (b) implements the generalized phase estimation measurement $M_{\rho}$:
\begin{align}
|\tau,1,1\rangle\otimes |\psi\rangle  \xrightarrow[]{\mathrm{QFT}_n^{\dagger}\otimes I} \frac{1}{\sqrt{n!}}&\sum_{\alpha\in S_n}|\alpha\rangle \otimes |\psi\rangle \\
\xrightarrow[]{C-\rho^{\dagger}} 
&\frac{1}{\sqrt{n!}}\sum_{\alpha\in S_n}|\alpha\rangle \otimes \rho^{\dagger}(\alpha)|\psi\rangle \\
 \xrightarrow[]{\text{Measure } M_L} 
&\frac{1}{\sqrt{n!}}\sum_{\alpha\in S_n}\Pi^L_{\omega}|\alpha\rangle \otimes \rho^{\dagger}(\alpha)|\psi\rangle \quad \quad (\text{conditioned on meas. outcome }\omega)\\
=&\frac{d_{\omega}}{(n!)^{3/2}}\sum_{\alpha,\sigma\in S_n} \chi_{\omega}(\sigma)|\sigma\alpha\rangle \otimes \rho^{\dagger}(\alpha)|\psi\rangle\\
\xrightarrow[]{C-\rho} & \frac{d_{\omega}}{(n!)^{3/2}}\sum_{\alpha,\sigma\in S_n} \chi_{\omega}(\sigma)|\sigma\alpha\rangle \otimes \rho(\sigma \alpha)\rho^{\dagger}(\alpha)|\psi\rangle.\label{eq:outputstate}
\end{align}
Using the fact that $\rho(\sigma\alpha)=\rho(\sigma)\rho(\alpha)$ and that $\rho^{\dagger}(\alpha)=\rho(\alpha)^{-1}$ we see that the output state Eq.~\eqref{eq:outputstate} is equal to
\begin{align}
 \frac{d_{\omega}}{(n!)^{3/2}}\sum_{\alpha,\sigma\in S_n} \chi_{\omega}(\sigma)|\sigma\alpha\rangle \otimes \rho(\sigma)|\psi\rangle
=\frac{1}{\sqrt{n!}} \sum_{\beta\in S_n}|\beta\rangle \otimes \frac{d_\omega}{n!}\sum_{\sigma\in S_n} \chi_{\omega}(\sigma) \rho(\sigma)|\psi\rangle.
\end{align}
On the right hand side the state of the second register is $\Pi_{\omega}|\psi\rangle$, as desired.

\section{$\QMA$, $\#\BQP$, and $\QAPC$}
Recall the definition of a verification circuit from the main text. In the following we write $p(\psi)$ for the acceptance probability of a verification circuit with input state $\psi$. 

A (promise) problem is in $\QMA$ \textit{with completeness $c$ and soundness $s$}, also denoted $\QMA(c,s)$, if there exists a uniform polynomial-sized family of verification circuits $C_x$ labeled by instances $x$ of the problem, such that (A) If $x$ is a yes instance then there exists a witness $\psi$ such that $p(\psi)\geq c$, and (B) If $x$ is a no instance then $p(\psi)\leq s$ for all  $\psi$.  

A standard convention is to define $\QMA\equiv \QMA(2/3,1/3)$. The definition is not very sensitive to the choice of completeness $c$ and soundness $s$: it is known that $\QMA(c,s)=\QMA(2/3,1/3)$ whenever $0<s<c\leq 1$ and $c-s=\Omega(1/\mathrm{poly}(n))$.  In fact, one can reduce the error of any $\QMA$ verifier exponentially in the input size $n$, so that, e.g.,  $\QMA(2/3,1/3)=\QMA(4^{-n},1-4^{-n})$ as well \cite{kitaev2002classical, marriott2005quantum}. 

We choose to define $\#\BQP$ as a class of relation problems rather than a class of functions, languages (decision problems), or promise problems. Recall from the main text that Eq.~\eqref{eq:POVM} describes the POVM element corresponding to the accept outcome of a quantum verification circuit and $N_q$ denotes the number of eigenvalues of $A$ that are at least $q$. We define $\#\BQP$ as the class of problems that can be expressed as: given a quantum verification circuit with an $n$-qubit input register, total size $\mathrm{poly}(n)$, and two thresholds $0<b<a\leq 1$ and $a-b=\Omega(1/\mathrm{poly}(n))$, compute an estimate $N$ such that $N_a\leq N\leq N_b$. This defines a relation problem because there may be more than one solution for a given instance (this happens whenever $N_a\neq N_b$).

 Ref.~\cite{brown2011computational} puts forward a slightly different definition of $\#\BQP$, as the class of promise problems obtained by augmenting and specializing the above definition with the promise that $N_a=N_b$. This has the pleasing feature that it results in a class of functions rather than relations, though at the expense of adding a promise.

One of the main results established in Ref.~\cite{brown2011computational} is that any problem in  $\#\BQP$ (as they define it) can be efficiently reduced to computing a $\#\P$ function. We now describe how their proof  can be straightforwardly adapted to our setting, establishing the same result for the definition of $\#\BQP$ that we use in this paper.  Using the in-place error reduction for $\QMA$ \cite{marriott2005quantum} we may assume without loss of generality that $a=1-2^{-n-2}$ and $b=2^{-n-2}$ (see, e.g. the proof of Lemma 8 in the supplementary material of Ref.~\cite{bravyi2022quantum}).
Then 
\begin{equation}
\mathrm{Tr}(A)\leq N_b+2^{-n-2}\cdot (2^n-N_b)\leq N_b+1/4,
\label{eq:trup}
\end{equation}
and
\begin{equation}
\mathrm{Tr}(A)\geq N_a (1-2^{-n-2})\geq N_a-1/4.
\label{eq:trlow}
\end{equation}
Eqs.~(\ref{eq:trup},\ref{eq:trlow}) show that the nearest integer $N$ to $\mathrm{Tr}(A)$ solves the given $\#\BQP$ problem as it satisfies $N_a\leq N\leq N_b$. It then suffices to show that we can compute $\mathrm{Tr}(A)$ to within an additive error of $1/4$ using efficient classical computation along with a single call to a $\#\P$ oracle. This is established in Theorem 14 of Ref.~\cite{brown2011computational} via a proof that is based on expressing $\mathrm{Tr}(A)$ as a sum over paths, see Ref.~\cite{brown2011computational} for details.

\section{Efficient sampling algorithm}
Now let us describe the algorithm stated in Lemma \ref{lem:samplekron}, which appears implicitly in Ref.~\cite{Moore07}. 

The completely mixed state over the image of the projector $\Pi^L_{\mu}$ can be prepared as follows.
\begin{align}
    \frac{\Pi^L_\mu}{d_\mu^2} &= \text{QFT}_n^\dag \left( \ket{\mu} \bra{\mu} \otimes I/d_{\mu} \otimes  I/d_{\mu}\right) \text{QFT}_n. 
\end{align}
From Eq.~\ref{eq:invf}: 
\begin{align}
    \text{QFT}_n^\dag \ket{\mu, i, j} &= \text{QFT}_n^\dag \sum_{\lambda \vdash n} \sum_{1 \leq k,l \leq d_\lambda} \underbrace{(\delta_{\lambda \mu} \delta_{ki} \delta_{kj})}_{c(\lambda)} \ket{\lambda, k, l} = \sum_\sigma \sqrt{\frac{d_\mu}{n!}} \rho_\mu^\dag(\sigma)_{ji} \ket{\sigma}.
\end{align}
There is a basis in which $\rho_\lambda$ are orthogonal matrices (Young's orthogonal form). It follows in this basis that $\rho_\mu^\dag(\sigma)_{ji} = \rho_\mu(\sigma)_{ij}$, so that we can write: 
\begin{align}
    \text{QFT}_n^\dag \ket{\mu, i, j} &= \sum_\sigma \sqrt{\frac{d_\mu}{n!}} \rho_\mu(\sigma)_{ij} \ket{\sigma},
\end{align}
and 
\begin{align}
   \frac{\Pi_\mu^L}{d_\mu^2} &= \frac{1}{d_\mu^2}\sum_{1 \leq i, j\leq d_\mu} \text{QFT}_n^\dag  \ket{\mu, i,j}\bra{\mu, i,j}\text{QFT}_n =  \sum_{\alpha, \beta \in S_n} \frac{1}{n! d_\mu} \chi_\mu(\alpha \beta^{-1}) \ket{\alpha} \bra{\beta},
\end{align}
We will need the projector formula: 
$\Pi_\lambda = {\frac{d_\lambda}{n!}} \sum_g \chi_\lambda(g) R(g) \otimes R(g)$, where $R$ is the regular representation. We already proved above how to implement this measurement using Generalized Phase Estimation. For completeness, we can verify the projection property of this operator as follows: 
\begin{align}
    \Pi_\lambda^2 &= \left(\frac{d_\lambda}{n!} \right)^2 \sum_{h \in S_n} \chi_\lambda(h) R(h) \otimes R(h)\sum_{g \in S_n} \chi_\lambda(g) R(g) \otimes R(g) \\ &= \left( \frac{d_\lambda}{n!} \right)^2 \sum_{g \in {S_n}} \sum_{g \in {S_n}} \chi_\lambda(h) \chi_\lambda(h^{-1}g) R(g) \otimes R(g) \\
    &= {\frac{d_\lambda}{n!}} \sum_{g \in S_n} \chi_\lambda(g) R(g) \otimes R(g),
\end{align}
where the last line follows from the convolution formula $\chi_\lambda * \chi_\lambda(s) = \sum_h \chi_\lambda(h) \chi_\lambda(h^{-1}s) = \frac{n!}{d_\lambda}\chi_\lambda(s)$. We now use that $\braket{a|R(g)|b} = \braket{a|gb} = \delta_{a,gb}$ to show:
\begin{align}
\frac{1}{d_\mu^2 d_\nu^2}
    \text{Tr}\left( (\Pi_\mu^L \otimes \Pi_\nu^L) \Pi_\lambda \right) &= \sum_{a,b \in S_n} \sum_{c,d \in S_n} \sum_{g \in S_n} \frac{d_\lambda}{n!^3 d_\mu d_\nu}\chi_\mu(ab^{-1}) \chi_\nu(cd^{-1}) \chi_\lambda(g) \text{Tr} \left( \ket{ac} \bra{bd} R(g) \otimes R(g) \right) \\
    &= \sum_{a,d,g \in S_n}\frac{d_\lambda}{n!^3 d_\mu d_\nu}\chi_\mu(g^{-1}) \chi_\nu(g) \chi_\lambda(g) = \sum_{g \in S_n} \frac{d_\lambda}{d_\mu d_\nu}\frac{ \chi_\mu(g^{-1}) \chi_\nu(g) \chi_\lambda(g)}{n!}.
\end{align}
Since $S_n$ has orthogonal representation for its matrices, we have that:
\begin{align}
    \chi_\mu(g^{-1}) = \text{Tr}(\rho_\lambda(g)^T) = \text{Tr}(\rho_\lambda(g)) =  \chi_\mu(g),
\end{align}
from which:  
\begin{align}
    p(\lambda) &= \text{Tr} \left(\frac{\Pi^L_\mu \otimes \Pi^L_\nu}{d_\mu^2 d_\nu^2} \Pi^L_\lambda \right)   =\sum_{g \in S_n} \frac{d_\lambda}{d_\mu d_\nu}\frac{ \chi_\mu(g) \chi_\nu(g) \chi_\lambda(g)}{n!} =  \frac{d_\lambda}{d_\mu d_\nu} g_{\mu \nu \lambda}.
\end{align}

\section{Row sums}

\begin{lemma}
The multiplicity of an irreducible representation $\rho_\lambda$ in the conjugation representation $\rho_c$ is given by the row sum $R_{\lambda}$ defined in Eq.~\eqref{eq:rowsum}.

\end{lemma}
\begin{proof}
Recall that the centralizer $Z(\pi)$ of an element $\pi \in S_n$ is defined as the set of elements $g \in S_n$ that commute with $\pi$. The conjugation representation acts on $\mathbb{C}S_n$ as
\[
\rho_c(\pi)|\sigma\rangle = |\pi \sigma \pi^{-1}\rangle 
\]
and therefore 
\[
\chi_{c}(\pi)\equiv \mathrm{Tr}(\rho_c(\pi))=|Z(\pi)|,
\]
where $Z(\pi)$ is the centralizer of $\pi \in S_n$. Note that all elements $\pi\in S_n$ in the same conjugacy class $C$ have the same centralizer which we denote $Z(C)$.

We can decompose the conjugation representation into irreducibles as $\rho_c \simeq \bigoplus_{\lambda \vdash n} m_{\lambda} \rho_{\lambda}$ for some multiplicities $m_\lambda$. Taking the trace gives 
\[
\chi_c(\pi)=\sum_{\lambda \vdash n} m_{\lambda}\chi_{\lambda}(\pi)
\]
for all $\pi\in S_n$. Using orthogonality of the irreducible characters we arrive at
\begin{align}
    m_\lambda &= \frac{1}{n!} \sum_{g \in G} \chi_\lambda(g) \chi_{c}(g) = \frac{1}{n!} \sum_{C \vdash n} \chi_\lambda(C) |C| |Z(C)| \label{eq:rowsumC},
\end{align}
where $|C|$ is the order of the conjugacy class $C$. By Lagrange's theorem, $n! / |C| = |Z(C)|$, and substituting this into Eq.~\eqref{eq:rowsumC} shows that $m_{\lambda}=R_{\lambda}$ and completes the proof.

\end{proof}
\begin{lemma}
$R_{\lambda}=d_{\lambda}^{-1}\mathrm{Tr}(\Pi_{\lambda}^{c})$
\end{lemma}
\begin{proof}
By definition, 
\[
\Pi_{\lambda}^{c}=\frac{d_{\lambda}}{n!}\sum_{g\in S_n}\chi_{\lambda}(g)\rho_c(g),
\]
and therefore
\[
\mathrm{Tr}(\Pi_{\lambda}^{c})=\frac{d_{\lambda}}{n!}\sum_{g\in S_n}\chi_{\lambda}(g)\chi_c(g) = d_\lambda R_\lambda,
\]
by Eq. \eqref{eq:rowsumC}
\end{proof}
\end{document}